\documentclass[a4paper,UKenglish,cleveref,hyperref,autoref]{lipics-v2021}

\title{Independent Set Hardness in Graphs of Bounded Twin-Width and Low-Radius Merge-Width}

\titlerunning{Independent Set Hardness in Graphs of Bounded Twin-Width and Merge-Width}

\author{\'{E}douard Bonnet}{Univ Lyon, CNRS, ENS de Lyon, Université Claude Bernard Lyon 1, LIP UMR5668, France \and \url{http://perso.ens-lyon.fr/edouard.bonnet/}}{edouard.bonnet@ens-lyon.fr}{https://orcid.org/0000-0002-1653-5822}{}

\author{Ma\"el Dumas}
       {Institute of Informatics, University of Warsaw, Poland}
       {mdumas@mimuw.edu.pl}
       {https://orcid.org/0009-0004-8906-9801}
       {}

\author{Julien Duron}{Institute of Informatics, University of Warsaw, Poland}{j.duron@uw.edu.pl}{}{}

\authorrunning{\'E. Bonnet, M. Dumas, J. Duron}

\Copyright{Édouard Bonnet, Maël Dumas,  Julien Duron}%TODO mandatory, please use full first names. LIPIcs license is "CC-BY";  http://creativecommons.org/licenses/by/3.0/

\category{}%optional, e.g. invited paper

\relatedversion{}%optional, e.g. full version hosted on arXiv, HAL, or other respository/website
\supplement{}%optional, e.g. related research data, source code, ... hosted on a repository like zenodo, figshare, GitHub, ...

\funding{JD received funding from ERC grant BUKA (No.~101126229).
MD and JD were supported by the project BOBR that has received funding from ERC
under the European Union's Horizon 2020 research and innovation programme,
grant agreement No.~948057.}%optional, to capture a funding statement, which applies to all authors. Please enter author specific funding statements as fifth argument of the \author macro.
 
\acknowledgements{We want to thank Karolina Drabik, Jakub Nowakowski, Nikolas Mählmann, and Szymon Toruńczyk for fruitful discussions at an early stage of this project.}%optional

\nolinenumbers %uncomment to disable line numbering

\hideLIPIcs  %uncomment to remove references to LIPIcs series (logo, DOI, ...), e.g. when preparing a pre-final version to be uploaded to arXiv or another public repository

\EventEditors{John Q. Open and Joan R. Access}
\EventNoEds{2}
\EventLongTitle{42nd Conference on Very Important Topics (CVIT 2016)}
\EventShortTitle{CVIT 2016}
\EventAcronym{CVIT}
\EventYear{2016}
\EventDate{December 24--27, 2016}
\EventLocation{Little Whinging, United Kingdom}
\EventLogo{}
\SeriesVolume{42}
\ArticleNo{23}

\usepackage[utf8]{inputenc}  

\usepackage[T1]{fontenc}
\usepackage{lmodern}

\usepackage[colorinlistoftodos,bordercolor=orange,backgroundcolor=orange!20,linecolor=orange,textsize=normalsize]{todonotes}

\usepackage{amsmath}  
\usepackage{amssymb}     
\usepackage{bbm}
\usepackage{accents}
\usepackage{complexity}
\usepackage{booktabs}
\usepackage{fixmath}
\usepackage{paralist}
\usepackage{bm}

\makeatletter
\newtheorem*{rep@theorem}{\rep@title}
\newcommand{\newreptheorem}[2]{%
\newenvironment{rep#1}[1]{%
 \def\rep@title{#2 \ref{##1}}%
 \begin{rep@theorem}}%
 {\end{rep@theorem}}}
\makeatother

\newreptheorem{theorem}{Theorem}
\newreptheorem{lemma}{Lemma}

\crefname{observation}{Observation}{Observations}

\usepackage{xspace}
\usepackage{tikz}

\usepackage[ruled,vlined,linesnumbered]{algorithm2e}

\usetikzlibrary{fit}
\usetikzlibrary{arrows}
\usetikzlibrary{patterns}
\usetikzlibrary{calc}
\usetikzlibrary{shapes}
\usetikzlibrary{positioning}
\usetikzlibrary{math}
\usetikzlibrary{shapes.symbols}
\usetikzlibrary{decorations.pathreplacing}
\tikzset{draw half paths/.style 2 args={%
  decoration={show path construction,
    lineto code={
      \draw [#1] (\tikzinputsegmentfirst) -- 
         ($(\tikzinputsegmentfirst)!0.5!(\tikzinputsegmentlast)$);
      \draw [#2] ($(\tikzinputsegmentfirst)!0.5!(\tikzinputsegmentlast)$)
        -- (\tikzinputsegmentlast);
    }
  }, decorate
}}

\usepackage[scr=boondox,scrscaled=1.05]{mathalfa}

\renewcommand{\geq}{\geqslant}
\renewcommand{\leq}{\leqslant}

\renewcommand{\le}{\leq}
\renewcommand{\ge}{\geq}

\usepackage{fontawesome5}

\newcommand{\aiproof}{%
  \raisebox{-0.15ex}{\faRobot}\,\(\vdash\)%
}

\newcommand{\val}{\text{val}}

\newcommand{\tsat}{\textsc{3-SAT}\xspace}

\newcommand{\kds}{\textsc{$k$-Dominating Set}\xspace}

\newcommand{\kis}{\textsc{$k$-Independent Set}\xspace}

\newcommand{\mis}{\textsc{Max Independent Set}\xspace}

\newcommand{\smis}{\textsc{MIS}\xspace}

\newcommand{\gridtiling}{\textsc{Grid Tiling}\xspace}

\theoremstyle{definition}

\newcommand{\tww}{\text{tww}}
\newcommand{\mw}{\text{mw}}
\renewcommand{\P}{{\mathcal P}}

\newcommand{\dist}{\mathrm{dist}}

\newcommand{\Ball}{\mathrm{Ball}}

\newcommand{\Oh}{\mathcal O}

\renewcommand{\S}{\mathcal S}
\renewcommand{\C}{\mathcal C}

\newcommand{\N}{\mathbb N}

\begin{document}

\maketitle

\begin{abstract}
  For every $\varepsilon > 0$, \textsc{Max Independent Set} admits a~polynomial-time $n^\varepsilon$-approximation algorithm on \mbox{$n$-vertex} graphs of effectively bounded twin-width [Bergé et al., STACS '23].
  The approximation factor actually obtained is more precisely $n^{\Oh(1/ \log \log n)}$.
  Prior to the current paper, no approximation hardness was known for this problem, and the existence of a~polynomial-time approximation scheme (PTAS) was repeatedly raised as an open question.
  We answer this question in a~strong sense: We show that there is a~constant $\gamma > 0$ such that a~polynomial-time $n^{\gamma/ (\log \log n)^2}$-approximation algorithm for \textsc{Max Independent Set} on graphs of twin-width at most 4 would refute the Exponential-Time Hypothesis (ETH).
  This lower bound further holds if a~4-sequence is provided as part of the input.  
  We show the same hardness of approximation for \textsc{Min Coloring}, which also has a~nearly matching $n^{\Oh(1/ \log \log n)}$-approximation algorithm on graphs of effectively bounded twin-width.

  We also clarify the parameterized complexity of \textsc{$k$-Independent Set} on graphs of bounded radius-$r$ merge-width when the range of~$r$ is limited.
  There is a~fixed-parameter tractable algorithm for \textsc{$k$-Independent Set} on graphs given with radius-$2^{\Oh(k^2)}$ merge sequences of bounded width [Dreier and Toruńczyk, STOC '25].
  We complement this result by showing that \textsc{$k$-Independent Set} is W[1]-hard on graphs given with radius-$o(k)$ merge sequences of bounded width.
  We further show that this result also holds for \textsc{$k$-Dominating Set}.
\end{abstract}

\section{Introduction}\label{sec:intro}

The problem of finding in a~graph a~largest subset of pairwise non-adjacent vertices, \mis (\smis for short), is NP-complete~\cite{GJ79}, and very hard to approximate: For any $\varepsilon > 0$, an $n^{1-\varepsilon}$-approximation algorithm would imply that P\,$=$\,NP~\cite{Hastad96,Zuckerman07}.
The best known polynomial-time approximation factor is~$\frac{n (\log \log n)^2}{\log^3 n}$~\cite{Feige04}.
In contrast, this central problem enjoys much better approximation algorithms on structured graph classes.
A~historical example is that of planar graphs where, due to Baker's shifting technique, a~polynomial-time approximation scheme (PTAS) exists~\cite{Baker94}.
\smis also has a~PTAS on minor-free classes~\cite{Grohe03} and on disk graphs~\cite{ChanH12}.
If a~geometric representation is provided, there is, for every $\varepsilon > 0$, an $n^\varepsilon$-approximation algorithm for the independence number of the intersection graph of $n$ curves every two of which have at~most a~constant number of intersections~\cite{FoxP11}, and a~quasipolynomial-time approximation scheme (QPTAS) for general string graphs~\cite{AdamaszekHW19}.

In this paper, our focus is on classes of bounded twin-width~\cite{twin-width1} and of bounded merge-width~\cite{DreierT25}.
These graph parameters were introduced in 2020 and 2025, respectively.
Classes of bounded twin-width include classes of bounded clique-width, classes excluding a~fixed minor, unit interval graphs, and $d$-dimensional grids~\cite{twin-width1}.
Classes of bounded merge-width are even more general and further include classes of bounded expansion~\cite{DreierT25}.
The main algorithmic application of classes of bounded twin-width and of bounded merge-width, when bounded-width witnesses are given, is that first-order model checking is fixed-parameter tractable~\cite{twin-width1,DreierT25}.

As we detail in the next paragraph, some improved approximation algorithms exist on classes of bounded twin-width.
In general, however, there is no precise understanding of the provably best algorithms achievable when twin-width or merge-width is bounded.
This is reflected in large approximability gaps and the lack of tight lower bounds under the Exponential-Time Hypothesis (ETH).
Our motivation is to fill those gaps.

In classes of effectively\footnote{See this and all relevant definitions in~\cref{sec:prelim}.} bounded twin-width, $n^\varepsilon$-approximation algorithms for \smis, \textsc{Min Coloring}, and other related problems were shown for every $\varepsilon > 0$~\cite{BergeBDW23}.
The proof in~\cite{BergeBDW23} reveals that the approximation factor for \smis is more specifically $n^{c/\log \log n}$, for some constant~$c$.
While it was known that \mis is NP-hard on graphs of bounded twin-width \cite{twin-width1}, even bounded by~4~\cite{Berge22} (also see~\cite{Bonnet26} for the status of other graph problems on graphs of low twin-width), no conditional hardness of approximation was established on classes of bounded twin-width. 
Even a~polynomial-time approximation scheme (PTAS) was not excluded.
Indeed, the existence of a~PTAS was raised several times \cite{twin-width2,twin-width3,BergeBDW23,Depres24,Bonnet24}. 
It was nonetheless observed that any $\Oh(1)$-approximation algorithm would imply the existence of a~PTAS~\cite{twin-width3}.

We rule this out under the ETH and show a~much tighter hardness-of-approximation result that nearly matches the $n^{c/\log \log n}$-approximation algorithm.

\begin{theorem}\label{thm:inapprox-main}
  Unless the ETH fails, there is a~constant $\gamma > 0$ such that \mis does not admit a~polynomial-time $n^{\gamma/(\log \log n)^2}$-approximation algorithm on $n$-vertex graphs of twin-width at~most~4, even when a~4-sequence is given as part of the input. 
\end{theorem}

The reduction uses, in order, three ingredients: (1) a~classical reduction from \tsat to \smis, (2) a~long even subdivision of the resulting instances, and (3) a~lexicographic power of the resulting graphs.

``Long'' subdivisions have twin-width at~most~4~\cite{Berge22}, which explains the role of (2). 
\emph{Even} subdivisions preserve the NP-hardness of \smis~\cite{Poljak74}.
However, ``long'' subdivisions do not preserve any hardness of approximation.
There is a~straightforward PTAS on these graphs: select every other vertex along each subdivided edge.
So, there is a~psychological barrier in using (2).
Nevertheless, with the right exponent, (3) resurrects this seemingly defunct approximation gap, while keeping the reduction within the desired running-time budget.

Similar ingredients can be adapted to give the same result for \textsc{Min Coloring}.

\begin{theorem}[\aiproof{}]\label{thm:inapprox-col}
  Unless the ETH fails, there is a~constant $\gamma > 0$ such that \textsc{Min Coloring} does not admit a~polynomial-time $n^{\gamma/(\log \log n)^2}$-approximation algorithm on \mbox{$n$-vertex} graphs of twin-width at~most~4, even when a~4-sequence is given as part of the input. 
\end{theorem}

\Cref{thm:inapprox-col} implies \cref{thm:inapprox-main}.
Indeed, a~polynomial-time \mbox{$r$-approximation} algorithm for \smis implies a~polynomial-time $\Oh(r \log n)$-approximation algorithm for \textsc{Min Coloring}.
Repeatedly removing an $r$-approximate maximum independent set and assigning it a~new color yields an $r(1+\ln n)$-approximate minimum coloring, by the classic analysis of \textsc{Set Cover}'s greedy algorithm.
Moreover, $n^{\Theta(1)/(\log \log n)^2} \cdot \Oh(\log n) = n^{\Theta(1)/(\log \log n)^2}$.

We keep both proofs because the proof of~\cref{thm:inapprox-main} implements the above outline without the extra technicalities needed for~\textsc{Min Coloring}.

\medskip

We then delineate the tractability boundary of \kis (\textsc{$k$-IS} for short) on graphs of bounded merge-width.
Strictly speaking, merge-width is a~\emph{family} of parameters.
For every positive integer $r$, the radius-$r$ merge-width, denoted by~$\mw_r$, must be bounded by a~function of~$r$, for the merge-width itself to be bounded.
Low-radius merge-width, think $\mw_1$, $\mw_2$, or $\mw_c$ for some constant~$c$, is of independent interest.
There are some properties that only require these parameters to be bounded.
Notably, classes of bounded merge-width were proven $\chi$-bounded by showing the stronger statement that classes of bounded~$\mw_2$ are $\chi$-bounded~\cite{Bonamy25} (and the $\chi$-boundedness of classes of bounded $\mw_1$ remains open).

Given radius-$2^{c k^2}$ merge sequences of bounded width\footnote{witnessing that the radius-$2^{c k^2}$ merge-width is bounded; see~\cref{subsec:merge-width} for the relevant definitions.} for some appropriate constant $c>0$, \textsc{$k$-IS} admits a~fixed-parameter tractable (FPT) algorithm~\cite{DreierT25}.
In fact, this is more generally the case for first-order model checking.
Dreier and Toruńczyk ask whether a~simpler and faster algorithm could exist for \kis (and \kds).

A~related question is whether these problems admit FPT algorithms given radius-$r$ merge sequences of bounded width already in the regime $r=\Oh(1)$.
We rule this out for any radius $r = o(k)$, with a~reduction from the W[1]-hard \gridtiling problem~\cite{Marx07} to~\textsc{$k$-IS}.
We ensure that the produced graphs have bounded radius-$r$ merge-width by spacing the \gridtiling cells out with ``paths of co-matchings'' of length $\Theta(r)$.

\begin{theorem}\label{thm:main-paramIS}
  $\forall r=o(k)$, \kis is W[1]-hard in graphs of bounded $\mw_r$.
\end{theorem}

We then adapt the previous reduction to get the same statement for \kds.

\begin{theorem}\label{thm:main-paramDS}
 $\forall r=o(k)$, \kds is W[1]-hard in graphs of bounded $\mw_r$.
\end{theorem}

\textbf{Future work.}
Our work naturally leads to the following questions.
Combining \cref{thm:inapprox-main,thm:inapprox-col} and~\cite{BergeBDW23}, we now know that, unless the ETH fails, the best approximation factor for \mis and \textsc{Min Coloring} lies somewhere between $n^{\Theta(1)/(\log \log n)^2}$ and $n^{\Theta(1)/\log \log n}$.
Can one obtain matching bounds, up to constants in the exponent?

Ruling out a~PTAS for \smis on graphs of bounded twin-width under the weaker assumption that P\,$\neq$\,NP remains open.
Let us also mention the question of obtaining nontrivial approximation algorithms for these problems when the contraction sequence is \emph{not} provided.

\cref{thm:main-paramIS,thm:main-paramDS} and \cite{DreierT25} raise the question of the smallest radius $r(k)$ such that \mbox{\kis} (and \textsc{$k$-Dominating Set}) is FPT on graphs given with a~radius-$r(k)$ merge sequence of bounded width.
We now know that this radius is between $\Omega(k)$ and $2^{\Oh(k^2)}$.

\medskip

\textbf{AI disclosure.}
After we proved \cref{thm:inapprox-main}, we suspected that the same result should hold for \textsc{Min Coloring}, and could probably be proven in a~similar fashion.
We delegated this adaptation to ChatGPT Pro 5.5 Extended Thinking.
We provided it with the proof of~\cref{thm:inapprox-main} with a~concise prompt asking for a~possible transfer to \textsc{Min Coloring}.
Its first response, produced after 17 minutes, pinpointed the places where the same proof would break.
We then clarified that the task was the adaptation.
After another 50 minutes, it provided a~complete proof of~\cref{thm:inapprox-col}.
Although it follows the proof of~\cref{thm:inapprox-main} at a~high level, this proof contains several new elements, notably the initial hardness result~\cite[Thm.\ 14]{Bogdanov02}, a~modified subdivision, and a~reduction to leverage the chromatic number lower bound that comes from an upper bound on the independence number.  
The write-up of this result, and of the paper as a~whole, is ours.

\section{Preliminaries}\label{sec:prelim}

For two integers $i$ and $j$, we denote by $[i,j]$ the set of integers that are at least $i$ and at most~$j$.
For every integer $i$, $[i]$ is a shorthand for $[1,i]$.
In this paper, the $\log$ function is in base 2.

The Exponential-Time Hypothesis (ETH) asserts that there is a~real constant $\lambda > 1$ such that $n$-variable \tsat cannot be solved in time $\Oh(\lambda^n)$.  

\subsection{Standard graph-theoretic definitions and notation}

We denote by $V(G)$ and $E(G)$ the set of vertices and edges of a graph $G$, respectively.
We denote by $\lVert G \rVert$ the total size of~$G$, that is, $|V(G)|+|E(G)|$.
For $S \subseteq V(G)$, the \emph{subgraph of $G$ induced by $S$}, denoted $G[S]$, is obtained by removing from $G$ all the vertices that are not in $S$.
%Then $G-S$ is a shorthand for $G[V(G)\setminus S]$.
We denote by $N_G(v)$ the set of neighbors of~$v$ in~$G$.
We denote by $\dist_G(u,v)$ the length of a shortest path between~$u$ and~$v$ in $G$, and set $\Ball^r_G(v) := \{w\in V(G) \mid \dist_G(v,w)\le r\}$.

A~\emph{subdivision} of a~graph $G$ is any graph $H$ obtained from $G$ by replacing edges $e$ of~$G$ by paths with at~least~one edge whose extremities are the endpoints of~$e$.
A~\emph{$(\geqslant s)$-subdivision} is a~subdivision where every edge is replaced by a~path with at~least~$s$ internal vertices. 

\subsection{Trigraphs, contraction sequences, and twin-width}

A~\emph{trigraph} $G$ has vertex set $V(G)$, black edge set $E(G)$, red edge set $R(G)$ such that $E(G) \cap R(G) = \emptyset$ (and $E(G),R(G) \subseteq \binom{V(G)}{2}$).
Two vertices $u, v$ such that $uv \in R(G)$ are called \emph{red neighbors}.
The \emph{red degree} of~$u$ is its number of red neighbors.
The \emph{maximum red degree} of~$G$ is the maximum red degree among all its vertices.

For distinct vertices $u,v \in V(G)$, the \emph{contraction} of $u$ and $v$ replaces them with a new vertex $w$.
For every $x \in V(G) \setminus \{u,v\}$, the pair $wx$ is a black edge if both $ux$ and $vx$ are black edges, it is a non-edge if both $ux$ and $vx$ are non-edges, and it is a red edge otherwise.
All pairs not incident with $u$ or $v$ remain unchanged.

A \emph{contraction sequence} of a graph $G$, viewed as a trigraph with
no red edges, is a sequence $G = G_n, G_{n-1}, \ldots, G_1$ such that $|V(G_i)| = i$ and $G_{i-1}$ is obtained from $G_i$ by contracting two vertices.
The sequence is a~\emph{$d$-sequence} if, for every $i \in [n]$, the trigraph $G_i$ has maximum red degree at~most~$d$.
The \emph{twin-width} of~$G$, denoted by $\tww(G)$, is the minimum integer $d$ for which $G$ admits a~$d$-sequence.

To show \cref{thm:inapprox-main,thm:inapprox-col}, we will rely on the fact that long-enough subdivisions of general graphs have bounded twin-width.

\begin{lemma}[Theorem 7 in~\cite{Berge22}]\label{lem:tww-subd}
  For every $n$-vertex graph $G$, the $2 \lceil \log n \rceil$-subdivision $H$ of~$G$ has twin-width at~most~4.
  Furthermore, a~4-sequence of~$H$ can be computed in $\Oh(\lVert H \rVert)$ time.
\end{lemma}
The algorithmic statement of~\cref{lem:tww-subd} is implicit in~\cite{Berge22} and follows from the proof.

We say that a~graph class $\mathcal C$ has \emph{effectively bounded twin-width} if there is a~polynomial-time algorithm $\mathcal A$ and a~fixed integer~$d$ such that on every graph $G \in \mathcal C$, $\mathcal A$ outputs a~$d$-sequence of~$G$.
We observe that most classes currently known to have bounded twin-width have effectively bounded twin-width. 

\subsection{Merge-width}\label{subsec:merge-width}

For two sets $X,Y$, we write $XY := \{xy \mid x\in X, y\in Y, x\neq y\}$, where $xy$ denotes the unordered pair $\{x,y\}$.
Let~$\P$ be a~partition of the vertices $V$ of a~graph~$G$.
For a~set $X \subseteq V$, we denote by $X/\P$ the set of parts of $\P$ intersecting~$X$. 
Let $R \subseteq \binom{V}{2}$ be a set of pairs of vertices of~$G$.
The distance between $x$ and $y$ in $(V,R)$ is denoted by~$\dist_R(x,y)$, and we define $\Ball_R^r(x) := \{y \in V \mid \dist_R(x,y) \leqslant r\}$.
The \emph{radius-$r$ width} of $(\P,R)$~is defined as
\[ \max_{v \in V} \left|\Ball^r_R(v)/\P\right|. \]

We say that~$\P$ is \emph{homogeneous modulo $R$} (in $G$) if for any parts $X,Y \in \P$ (possibly $X=Y$)
and pairs $xy,x'y'\in XY\setminus R$, we have that $xy \in E(G)$ if and only if $x'y' \in E(G)$.
A~partition $\P$ is \emph{coarser} than a~partition $\P'$, denoted by $\P' \preccurlyeq \P$, if every part of $\P'$ is contained in a~part of~$\P$.

\begin{definition}[Definition 3.1 in~\cite{DreierT25}]
A \emph{merge sequence} for a graph~$G$ is a~sequence $(\P_1,R_1),\dots,(\P_m,R_m)$ such that:
\begin{enumerate}
  \item $\P_1\preccurlyeq \P_2\preccurlyeq\ldots\preccurlyeq \P_m$ is a~sequence of ever coarser partitions of $V(G)$, with~$\P_1$ the partition into singletons and $\P_m$ the partition with one part, 
  \item $R_1 \subseteq \dots \subseteq R_m \subseteq \binom{V(G)}{2}$ is a monotone sequence of sets of pairs of vertices, and
  \item $\P_t$ is homogeneous modulo $R_t$, for $t=1,\ldots,m$.
\end{enumerate}
The \emph{radius-$r$ width} of this merge sequence is the maximum radius-$r$ width of $(\P_{t},R_{t+1})$,
for  $1 \le t < m$.
Finally, the \emph{radius-$r$ merge-width} of~$G$, denoted by~$\mw_r(G)$, is the minimum radius-$r$ width of a merge sequence for~$G$.
\end{definition}
The elements of $R_t$ are called \emph{resolved pairs} (at step~$t$).
Note that the partition $\P_{t+1}$ is obtained by merging some parts of~$\P_t$.
Note also that the index mismatch in $(\P_t,R_{t+1})$ is intentional, and prevents one from merging many parts and adding many resolved pairs all at once when going from step~$t$ to~$t+1$.

A graph class~$\C$ has \emph{bounded merge-width} if there is a function $f$ such that for every $G\in \C$ and every positive integer~$r$, we have $\mw_r(G) \leqslant f(r)$. 

%Graph classes of bounded twin-width have bounded merge-width~\cite{DreierT25}. 
%This can be easily seen, by considering a contraction sequence as a maximal sequence of partition $\P_1\prec\P_2\prec\ldots\prec \P_n$, 
%and for each partition $\P_i$, the corresponding set of resolved pair $R_i$ is the union of all $XY$ for pairs of non-homogeneous parts $X$ and $Y$ of $\P_i$ (which would imply a red edge in the corresponding trigraph).
%The radius-$r$ width of the resulting merge sequence is at most $\Oh(d^r)$ where $d$ is the twin-width of the graph.

\medskip

\textbf{Organization.}
We prove \cref{thm:inapprox-main} in \cref{sec:inapprox}, \cref{thm:inapprox-col} in \cref{sec:inapproxCOL}, \cref{thm:main-paramIS} in \cref{sec:w1-hard}, and  \cref{thm:main-paramDS} in \cref{sec:w1-hardDS}.

\section{Approximation Hardness of MIS in Graphs of Bounded Twin-Width}\label{sec:inapprox}

In this paper, all \tsat instances have only 3-clauses.
Equal literals within the same clause are allowed, but we disallow opposite literals of the same variable within the same clause.
A~\tsat-$B$ formula has no variable occurring (positively or negatively) more than $B$ times. 
The following theorem is a~consequence of Dinur's almost-linear PCP~\cite{Dinur07} and some folklore reductions; see, for instance,~\cite[Thm.\ 7]{Manurangsi17}.

\begin{theorem}[consequence of~Theorem 8.3 in~\cite{Dinur07}]\label{thm:sat-inapprox}
  Unless the ETH fails, there are positive constants $\beta < 1, B, c$ such that distinguishing $n$-variable $m$-clause \tsat-$B$ satisfiable formulas from instances where at~most $\beta m$ clauses can be satisfied requires $2^{\Omega(n/\log^c n)}$ time.    
\end{theorem}

We denote by $e(\varphi)$ the total number of unordered pairs of opposing literals appearing in~$\varphi$.
For instance, if $\varphi = x_1 \lor \neg x_2 \lor x_4, \neg x_1 \lor \neg x_1 \lor x_2, \neg x_2 \lor x_3 \lor \neg x_4$, we have $e(\varphi)=2+2+0+1=5$.
We denote by $\val(\varphi)$ the largest number of clauses satisfied by any truth assignment.
In particular, if $\varphi$ is satisfiable, then $\val(\varphi)$ is equal to its number of clauses.

\begin{theorem}\label{thm:log-n-subdivision}
  There is an $\Oh((m+e(\varphi)) \log m)$-time algorithm that, given an $m$-clause \tsat formula $\varphi$, outputs a~graph $G_\varphi$ of twin-width at~most~4 such that
  \begin{compactitem}
    \item $|V(G_\varphi)|=3m+(e(\varphi)+3m) \cdot 2 \lceil \log(3m) \rceil$, and 
    \item $\alpha(G_\varphi) = \val(\varphi) + (e(\varphi)+3m) \cdot \lceil \log(3m) \rceil$.
  \end{compactitem}
  Furthermore, a~4-sequence of $G_\varphi$ can be computed in $\Oh(\lVert G_\varphi \rVert)$ time.
\end{theorem}

\begin{proof}
  Let $J_\varphi$ be the graph with one vertex per literal occurrence in a~clause of~$\varphi$, and an edge between every two literals that are either in the same clause or that are of the same variable with opposite signs.
  In particular, $J_\varphi$ has exactly $3m$ vertices and exactly $3m+e(\varphi)$ edges (recall our definition of \tsat).

  Then, $G_\varphi$ is the $s$-subdivision of $J_\varphi$ with $s := 2 \lceil \log(3m) \rceil$.
  Note that $G_\varphi$ can be built from $\varphi$ in $\Oh((m+e(\varphi)) \log m)$ time.
  By~\cref{lem:tww-subd}, $G_\varphi$ has twin-width at~most~4, and a~4-sequence of $G_\varphi$ can be computed in $\Oh(\lVert G_\varphi \rVert)$ time.
  Since $G_\varphi$ is a~$s$-subdivision of a~$3m$-vertex $(3m+e(\varphi))$-edge graph, $|V(G_\varphi)|=3m+(e(\varphi)+3m) \cdot s$.

  The second item dates back to~Poljak~\cite{Poljak74}.
  We sketch this classic fact.

  The inequality $\alpha(G_\varphi) \geqslant \val(\varphi) + (e(\varphi)+3m) \cdot \frac{s}{2}$ holds, as there is an independent set of~$G_\varphi$ with one vertex corresponding to a~satisfied literal for each clause satisfied by a~fixed truth assignment that satisfies $\val(\varphi)$ clauses of~$\varphi$, and $\frac{s}{2}$ vertices per subdivided edge of~$J_\varphi$.

  The inequality $\alpha(G_\varphi) \leqslant \val(\varphi) + (e(\varphi)+3m) \cdot \frac{s}{2}$ holds since a~larger independent set of~$G_\varphi$ would contain more than $\val(\varphi)$ original vertices (those of~$J_\varphi$).
  This would imply that~$J_\varphi$ has an independent set of size greater than~$\val(\varphi)$.
  (We elaborate on this in the proof of~\cref{thm:log-n-triangle-subdivision}.)
  The (consistent, partial) truth assignment setting the corresponding literals to true would contradict the definition of~$\val(\varphi)$. 
\end{proof}

Given a~graph $G$ and a~positive integer~$k$, we denote by $G^{[k]}$ the \emph{$k$-th lexicographic power} of~$G$ with $V(G^{[k]})=V(G)^k$ and $(u_1, \ldots, u_k)(v_1, \ldots, v_k) \in E(G^{[k]})$ whenever there is some $i \in [k]$ such that $u_iv_i \in E(G)$ and for every $h \in [i-1]$, $u_h = v_h$.
The following is a~folklore observation.
We give a~recent reference for completeness.

\begin{lemma}[see Lemma 7.12 in~\cite{twin-width3}]\label{obs:alpha-lex-pow}
  For every graph $G$ and integer $k \geqslant 1$, $\alpha(G^{[k]})=\alpha(G)^k$.
  
  Furthermore, there is an algorithm that, given an independent set $S$ of $G^{[k]}$, outputs an independent set of~$G$ of size at~least $|S|^{1/k}$ in $\Oh(\lVert G^{[k]} \rVert)$ time.
\end{lemma}

The twin-width of every lexicographic power of a~graph $G$ is equal to that of~$G$.

\begin{lemma}[Lemma 7.11 in~\cite{twin-width3}]\label{lem:tww-lex-pow}
  For every graph $G$ and integer $k \geqslant 1$, $\tww(G^{[k]})=\tww(G)$.

  Furthermore, given a~$d$-sequence of~$G$, one can compute a~$d$-sequence of~$G^{[k]}$ in time $\Oh(|V(G^{[k]})|)$.
\end{lemma}

We now have all the ingredients to show our main result, which we restate for convenience.

\begin{reptheorem}{thm:inapprox-main}
  Unless the ETH fails, there is a~constant $\gamma > 0$ such that \mis does not admit a~polynomial-time $n^{\gamma/(\log \log n)^2}$-approximation algorithm on $n$-vertex graphs of twin-width at~most~4, even when 4-sequences are given as part of the input. 
\end{reptheorem}
\begin{proof}
  Fix the constants $\beta < 1, B, c$ so as to satisfy \cref{thm:sat-inapprox}.
  We make the following reduction.
  Every $N$-variable $m$-clause \tsat-$B$ formula $\varphi$ is mapped to the graph $H := G_\varphi^{[k]}$ with $k := \lceil m/\log^{c+2} m \rceil$.
  We will assume that $m$ is large enough.
  In particular, the following inequalities hold:
  \begin{enumerate}
  \item\label{ineq:1} $\log\left(4(3B^2+6) \log(3m)\right) < \log(3m)$;
  \item\label{ineq:2} $\log m > 2 \log \log^{c+2} m + \log 3$;
  \item\label{ineq:3} $\frac{(e(\varphi)+3m) \cdot \lceil \log(3m) \rceil}{(1-\beta)m} > 2$.
  \end{enumerate}
  
  Building $H$ takes $\Oh(m \log m) + (m \log m)^{\Oh(k)}$ time. 
  Indeed, note that $e(\varphi) \leqslant B^2 N \leqslant 3 B^2 m$.
  So the first term in the running time corresponds to building~$G_\varphi$, and the second term, to building its $k$-th lexicographic power~$H$.
  By~\cref{thm:log-n-subdivision,lem:tww-lex-pow}, $H$ indeed has twin-width at~most~4, and a~4-sequence of~$H$ can be computed in time polynomial in $\lVert H \rVert$.

  Suppose there is an~$n^{\Oh(1)}$-time $r$-approximation algorithm $\mathcal A$ for \mis on $n$-vertex graphs of twin-width at~most~4.
  Our proof will work even if $\mathcal A$ has access to a~4-sequence of~$H$.
  
  Let $\delta \leqslant \Delta := 6(B^2+1)$ be such that \[(e(\varphi)+3m) \cdot \lceil \log(3m) \rceil = \delta m \log(3m).\]
  We set the constant \[\gamma := \frac{0.8(1-\beta)}{8(\beta+\Delta)}.\]
  To conclude the proof, we assume that $r = n^{\gamma/(\log \log n)^2}$ with $n := |V(G_\varphi)|^k$, and proceed to refute the ETH.  
  
  When run on~$H$, $\mathcal A$ returns an independent set $I_H$ of~$H$ of size at~least $\alpha(H)/r$.
  By~\cref{obs:alpha-lex-pow}, we find in time $\Oh((m \log m)^{2k})$ an independent set $I$ of~$G_\varphi$ of size at~least~\[\left(\frac{\alpha(H)}{r}\right)^{1/k}=\frac{\alpha(G_\varphi)}{r^{1/k}}.\]
  We report that $\varphi$ is satisfiable if $|I| > \beta m + \delta m \log(3m)$, and that at most $\beta m$ clauses of $\varphi$ are satisfiable, otherwise.

  \medskip

  \textbf{Running time.} Let us first verify that the overall process takes $2^{o(N/\log^c N)}$.
  Building~$H$, optionally computing a~4-sequence of~$H$, running~$\mathcal A$ on~$H$, finding~$I$, and concluding whether $\varphi$ is satisfiable takes time
  \[(m \log m)^{\Oh(k)} + n^{\Oh(1)} + \Oh((m \log m)^{2k}) = m^{\Oh(m/\log^{c+2} m)} = 2^{\Oh(m/\log^{c+1} m)}\]
  and, since $m \leqslant B N$ and $m \geqslant N/3$,
  \[2^{\Oh(m/\log^{c+1} m)} = 2^{\Oh(B N/\log^{c+1}(N/3))} = 2^{o(N/\log^c N)}.\]

  \textbf{Correctness.} We now check that we are correct on the satisfiability of~$\varphi$, and thus conclude by~\cref{thm:sat-inapprox}.
  The crux of the argument is the following claim.

  \begin{claim}\label{clm:ratio-tracing}
    $\left(m + \delta m \log(3m) \right) \cdot \frac{1}{r^{1/k}} > \beta m + \delta m \log(3m)$.
  \end{claim}
  \begin{claimproof}
    We first recall that
    \[n = |V(G_\varphi)|^k=(3m+(e(\varphi)+3m) \cdot 2 \lceil \log(3m) \rceil)^k \leqslant \left(4(3B^2+6)m \log(3m)\right)^k.\]
    So,
    \begin{equation}\label{ineq:4}
      \log n \leqslant k \log m + k \log\left(4(3B^2+6) \log(3m)\right) < 2k \log(3m),
    \end{equation}
    where the last inequality holds by~\cref{ineq:1}.
    And \[\log \log n \geqslant \log k \geqslant \log m - \log \log^{c+2} m > \frac{1}{2} \log(3m),\]
    where the last inequality holds by~\cref{ineq:2}.
    Hence,
    \begin{equation}\label{ineq:5}
      (\log \log n)^2 > \frac{1}{4} \log^2(3m).
    \end{equation}

    On the one hand, 
    \[\log(r^{1/k}) = \log\left(\left(2^{\gamma \log n/(\log \log n)^2}\right)^{1/k}\right) = \frac{\gamma \log n}{k (\log \log n)^2} < \frac{\gamma \cdot 2 \log(3m)}{\frac{1}{4} \log^2(3m)} = \frac{0.8(1-\beta)}{(\beta+\Delta) \log(3m)},\]
    where the inequality comes from the upper bound on the numerator by~(\ref{ineq:4}) and the lower bound on the denominator by~(\ref{ineq:5}).

    On the other hand,
    \[\log \left( \frac{m + \delta m \log(3m)}{\beta m + \delta m \log(3m)} \right)
    = \log \left( 1+ \frac{(1-\beta)m}{\beta m + \delta m \log(3m)} \right) \]
    \[= \log \left(1+ \frac{1-\beta}{\beta + \delta \log(3m)}\right) > \frac{0.8(1-\beta)}{\beta + \delta \log(3m)} > \frac{0.8(1-\beta)}{(\beta + \delta) \log(3m)} \geqslant \frac{0.8(1-\beta)}{(\beta + \Delta) \log(3m)},\]
    where the first inequality holds because for every $x \in (0,1/2)$, $\log(1+x) > 0.8 x$, and by~\cref{ineq:3} $(1-\beta)/(\beta + \delta \log(3m)) \in (0,1/2)$, and the second inequality holds because $\log(3m) > 1$.  

    Therefore, \[ \frac{m + \delta m \log(3m)}{\beta m + \delta m \log(3m)} > r^{1/k},\]
    and we conclude.
  \end{claimproof}

  By~\cref{clm:ratio-tracing}, if $|I| \leqslant \beta m + \delta m \log(3m)$, then $\alpha(G_\varphi) < m + \delta m \log(3m)$.
  By~\cref{thm:log-n-subdivision}, this indeed implies that $\varphi$ is unsatisfiable.

  If instead, $|I| > \beta m + \delta m \log(3m)$, there is a~truth assignment that satisfies strictly more than $\beta m$ clauses of~$\varphi$.
  So, in this promise-problem setting, we correctly concluded that $\varphi$ is satisfiable.
\end{proof}

\section{Coloring Approximation Hardness in Graphs of Bounded Twin-Width}\label{sec:inapproxCOL}

We show the same hardness of approximation for \textsc{Min Coloring} as we showed in the previous section.
The general strategy is the same but several elements need to be changed.
The following consequence of \cite[Thm.\ 14]{Bogdanov02} and \cref{thm:sat-inapprox} is a~convenient alternative to the mere \cref{thm:sat-inapprox} of~\cref{sec:inapprox}. 

\begin{theorem}\label{thm:inapprox-max-3-col}
  There are positive constants $\beta < 1, \Delta, c$ such that, for $n$-vertex graphs $H$ of maximum degree at~most~$\Delta$ promised to satisfy one of the following two conditions  
  \begin{compactitem}
  \item $H$ is 3-colorable, or
  \item no induced subgraph of~$H$ on at~least~$\beta n$ vertices is 3-colorable
  \end{compactitem}
  determining which one holds requires $2^{\Omega(n/\log^c n)}$ time, unless the ETH fails. 
\end{theorem}

\begin{proof}
  We run the linear reduction from \tsat-$B$ to \textsc{3-Coloring} of~\cite[Thm.\ 14]{Bogdanov02} on the hard instances of~\cref{thm:sat-inapprox}.
  To avoid a~notational clash, we rename the constant $\beta$ as $\beta'<1$ in the latter theorem, and denote by $N$ the number of variables of the \tsat-$B$ instances, and by~$m$ the number of clauses in each instance. 

  The reduction of \cite{Bogdanov02} (see the description preceding Theorem 14) produces graphs with at most $\kappa N$ vertices and maximum degree at~most~$\Delta$ for some constants $\kappa := \kappa(B), \Delta$.
  The image of every satisfiable \tsat-$B$ instance is 3-colorable, whereas in every (improper) 3-coloring of the image $H$ of any \tsat-$B$ instance $\varphi$ such that $\val(\varphi) \leqslant \beta' m$ at~most a~$\beta''$ fraction of the edges are properly colored (i.e., bichromatic) where $\beta'' := \frac{\beta'+7}{8}$.
  We observe that $\beta'' < 1$, since $\beta' < 1$, and choose $\beta := 1-\frac{1-\beta''}{2\Delta}<1$.
  Let $n := |V(H)|$.

  Assume that there is a~3-colorable induced subgraph of~$H$ with at~least $\beta n$ vertices.
  Any extension to a~complete 3-coloring of~$H$ improperly colors at~most $(1-\beta)\Delta n < (1-\beta'')n$ edges, because $H$ has maximum degree at~most~$\Delta$.
  Here we observe that the reduction in~\cite{Bogdanov02} satisfies that $|V(H)| \leqslant |E(H)|$.
  Thus we reach the contradiction that more than a~$\beta''$ fraction of the edges of~$H$ are properly colored.

  This way, the two items of the current theorem correspond to the two outcomes of~\cref{thm:sat-inapprox}.
  We conclude by taking the same value for $c$ as in~\cref{thm:sat-inapprox} since the blow-up $\kappa$ is constant. 
\end{proof}

We now adapt~\cref{thm:log-n-subdivision}.
For that, we modify the subdivisions slightly without changing their twin-width upper bound of~4.
One should also notice the hybrid nature of the last two items: in ``positive'' instances, we preserve a~low \emph{chromatic number}, whereas in ``negative'' instances, we upper-bound the \emph{independence number} of the produced graphs.  

\begin{theorem}\label{thm:log-n-triangle-subdivision}
  There is an $\Oh(\lVert G \rVert)$-time algorithm that, given an $n$-vertex $m$-edge graph~$H$, outputs a~graph $G := G_H$ of twin-width at~most~4 such that
  \begin{compactitem}
    \item $|V(G)|=3n+(3n+3m) \cdot 3 \lceil \log(3n) \rceil$;
    \item if $H$ is 3-colorable, then $G$ is 3-colorable;
    \item for any number $p$, if no induced subgraph of~$H$ on at~least~$p$ vertices is 3-colorable, then \[\alpha(G) \leqslant p + (3n+3m) \cdot \lceil \log(3n) \rceil.\]
  \end{compactitem}
  Furthermore, a~4-sequence of $G$ can be computed in $\Oh(\lVert G \rVert)$ time.
\end{theorem}

\begin{proof}
  We first turn every vertex $v$ of~$H$ into a~triangle $(v,1)(v,2)(v,3)$.
  For every $uv \in E(H)$, we add the three edges $(u,1)(v,1)$, $(u,2)(v,2)$, $(u,3)(v,3)$.
  So far, we get a~graph $J$ with $3n$ vertices and $3n+3m$ edges.

  To obtain $G$, we $s$-subdivide every edge $e$ of $J$, for $s := 2 \lceil \log(3n) \rceil$.
  We arbitrarily orient $e=uv$, say, from~$u$ to~$v$, and we add a~true twin to each odd-indexed vertex on the new path between $u$ and $v$: $u (uv)_1 (uv)_2 (uv)_3 \ldots (uv)_s v$.
  (Note that this ``triangle-subdivision'' was already used in a~similar context in~\cite{Bonnet26}.)
  That concludes the construction of~$G$.
  
  From this description, one can build $G$ in time proportional to its size.
  One can then observe that $|V(G)|=3n+(3n+3m) \cdot 3 \lceil \log(3n) \rceil$.
  Indeed we add $\lceil \log(3n) \rceil$ true twins per subdivided edge of~$J$.
  The $\Oh(\lVert G \rVert)$-time algorithm to compute a~4-sequence of~$G$ follows the same argument as in the proof of~\cref{thm:log-n-subdivision} after contracting every pair of true twins.
  The latter has the same effect as removing the added true twins (without creating any red edges at this point).

  Let us show that $\chi(H) \leqslant 3$ implies that $\chi(G) \leqslant 3$.
  Let $\mu \colon V(H) \to [3]$ be a~proper \mbox{3-coloring} of $H$.
  We define $\nu' \colon V(J) \to [3]$ by $\nu'(u,i) := (\mu(u)+i) \bmod 3 + 1$.
  This is a~proper coloring of~$J$: each triangle $(u,1)(u,2)(u,3)$ has all three colors, and a~monochromatic edge $(u,i)(v,i)$ under $\nu'$ would imply that $uv$ is monochromatic under $\mu$.
  To extend $\nu' \colon V(J) \to [3]$ to a~proper coloring $\nu \colon V(G) \to [3]$ of~$G$, give the color $\nu'(u)$ to every even-indexed vertex on the path: $u (uv)_1 (uv)_2 (uv)_3 \ldots (uv)_s v$ (when $uv$ was oriented from $u$ to~$v$).
  Finally, we assign the two remaining colors (the ones not already in their neighborhood) to each pair of true twins.

  We move to the last item of the statement.
  First note that the set $T_{uv}$ of $3 \lceil \log(3n) \rceil$ new vertices introduced for each edge $uv \in E(J)$ can be vertex-partitioned into $\lceil \log(3n) \rceil$ triangles.
  Therefore, at~most $\lceil \log(3n) \rceil$ of these vertices can be part of the same independent set.
  Similarly to the mere subdivision, there is a~largest independent set of~$G$ that for every $uv \in E(J)$ contains at~most~one original vertex among $u, v$.
  Indeed, if both $u,v \in V(G)$ are in an independent set $I$ of~$G$, only $\lceil \log(3n) \rceil - 1$ vertices of~$T_{uv}$ can be part of~$I$.
  This is because $T_{uv} \setminus N_G(\{u,v\})$ can be vertex-partitioned into $\lceil \log(3n) \rceil - 1$ triangles.
  One can thus remove, say, $v$ from $I$, and have $\lceil \log(3n) \rceil$ vertices of~$T_{uv}$ in~$I$.
  We conclude that \[\alpha(G) \leqslant \alpha(J) + |E(J)| \lceil \log(3n) \rceil = \alpha(J) + (3n+3m) \lceil \log(3n) \rceil.\]
  Actually the equality holds.

  We then claim that $\alpha(J) \leqslant \max\{|X| : X \subseteq V(H), H[X]~\text{is 3-colorable}\}$ (again, the equality holds).
  Let $I$ be an independent set of~$J$.
  Because of the triangles $(u,1)(u,2)(u,3)$, $I$ contains at~most~one vertex of the form $(u,i)$ for every $u \in V(H)$.
  Furthermore, for every edge $uv \in E(H)$, it cannot be that both $(u,i)$ and $(v,i)$ are in~$I$ for some $i \in [3]$.
  Thus, the partial 3-coloring obtained by giving color $i$ to vertex $u$ for every $(u,i) \in I$ is well-defined, proper, and colors $|I|$ vertices of~$H$.  
  Therefore if no induced subgraph of~$H$ on at~least~$p$ vertices is 3-colorable, then \[\alpha(G) \leqslant p + (3n+3m) \cdot \lceil \log(3n) \rceil.\qedhere\]
\end{proof}

Using these ingredients, we can now prove the main result of this section.
The omitted proof is similar to the one of \cref{thm:inapprox-main}.
\begin{reptheorem}{thm:inapprox-col}
  Unless the ETH fails, there is a~constant $\gamma > 0$ such that \textsc{Min Coloring} does not admit a~polynomial-time $n^{\gamma/(\log \log n)^2}$-approximation algorithm on $n$-vertex graphs of twin-width at~most~4, even when a~4-sequence is given as part of the input. 
\end{reptheorem}

\begin{proof}
  Fix the constants $\beta < 1, \Delta, c$ satisfying \cref{thm:inapprox-max-3-col}.
  Every promised $N$-vertex instance $H$ of~\cref{thm:inapprox-max-3-col} is mapped to the graph $G^{[k]}$ with $G := G_H$, as defined in~\cref{thm:log-n-triangle-subdivision}, and $k := \lceil N/\log^{c+2} N \rceil$.
  We assume, without loss of generality, that $N \geqslant 2$ is large enough.
  In particular, the following inequalities hold:
  \begin{enumerate}
  \item\label{ineq:b1} $\lceil \log(3N) \rceil \leqslant 4 \log N$ (this is implied by $N \geqslant 2$);
  \item\label{ineq:b2} $(1-\beta)/(\beta + 3(\Delta+1) \lceil \log(3N) \rceil) \in (0,1/2)$;
  \item\label{ineq:b3} $\log k \geqslant \log \sqrt{N} = 0.5 \log N$;
  \item\label{ineq:b4} $12(\Delta+1)N \geqslant \lceil \log(3N) \rceil$.
  \item\label{ineq:b5} $N > \log^{c+2} N$, thus $k \geqslant 2$.
  \end{enumerate}

  Building $G$ takes $\Oh(\lVert G \rVert)=\Oh(\Delta N \log N)=\Oh(N \log N)$ time, and building $G^{[k]}$ takes time
  \[\Oh(\lVert G^{[k]} \rVert)=2^{\Oh\left(\frac{N \log N}{\log^{c+2} N}\right)} = 2^{\Oh\left(\frac{N}{\log^{c+1} N}\right)}.\]
  By~\cref{thm:log-n-triangle-subdivision,lem:tww-lex-pow}, $G^{[k]}$ indeed has twin-width at~most~4, and a~4-sequence of~$G^{[k]}$ can be computed in time polynomial in $\lVert G^{[k]} \rVert$.
  Thus the rest of the proof works even if the algorithm $\mathcal A$ has access to a~4-sequence of~$G^{[k]}$.

  Suppose there is an~$n^{\Oh(1)}$-time $r$-approximation algorithm $\mathcal A$ for \textsc{Min Coloring} on $n$-vertex graphs of twin-width at~most~4.
  We set $n := |V(G^{[k]})|$ and \[\gamma := \frac{0.8 (1-\beta)}{768(\beta + 3(\Delta+1))(\Delta+1)}.\]
  We now assume that $r = n^{\gamma/(\log \log n)^2}$ and refute the ETH.
  We run $\mathcal A$ on $G^{[k]}$ and obtain a~proper coloring of this graph with at~most $r \cdot \chi(G^{[k]})$ colors.
  Based on this coloring, we will see that we can decide whether $H$ is 3-colorable.
  This is, by~\cref{thm:inapprox-max-3-col}, a~contradiction to the ETH once we bound the running time of the reduction, which we now do.

  \medskip

  \textbf{Running time.} 
  Building~$G^{[k]}$, optionally computing a~4-sequence $\mathcal S$ of~$G^{[k]}$, and running~$\mathcal A$ on $G^{[k]}, \mathcal S$ takes time
  \[\Oh\left(2^{\Oh\left(\frac{N}{\log^{c+1} N}\right)} + n^{\Oh(1)}\right) = 2^{\Oh\left(\frac{N}{\log^{c+1} N}\right)}=2^{o(N/\log^c N)},\]
  as desired.

  \medskip

  \textbf{Case when $\bm{H}$ is 3-colorable.}
  By~\cref{thm:log-n-triangle-subdivision}, it holds that $\chi(G) \leqslant 3$.
  Therefore, $\chi(G^{[k]}) \leqslant 3^k$, by giving color $(\nu(v_1), \ldots, \nu(v_k))$ to each vertex $(v_1, \ldots, v_k) \in V(G^{[k]})$ for a~fixed proper 3-coloring $\nu$ of~$G$. 

  \medskip

  \textbf{Case when $\bm{H}$ has no 3-colorable induced subgraph on at least $\bm{\beta N}$ vertices.}
  We have 
  \[\chi(G^{[k]}) \geqslant \frac{|V(G^{[k]})|}{\alpha(G^{[k]})}= \frac{|V(G)|^k}{\alpha(G^{[k]})} \geqslant \frac{|V(G)|^k}{\alpha(G)^k}=\left(\frac{|V(G)|}{\alpha(G)}\right)^k.\]
  where the first inequality holds for any graph, and the second inequality uses \cref{obs:alpha-lex-pow}.
 
  We set \[q := (3N+3|E(H)|) \lceil \log(3N) \rceil \leqslant 3(\Delta+1)N \lceil \log(3N) \rceil,\] since $H$ has maximum degree at~most~$\Delta$.
  By~\cref{thm:log-n-triangle-subdivision},
  \[\alpha(G) \leqslant \beta N + q \text{~and~} |V(G)| = 3N + 3q.\]
  Thus,
  \[\chi(G^{[k]}) \geqslant \left(\frac{3N+3q}{\beta N+q}\right)^k = 3^k \cdot \left(1+\frac{(1-\beta)N}{\beta N+q}\right)^k \geqslant 3^k \cdot \left(1+\frac{1-\beta}{\beta + 3(\Delta+1) \lceil \log(3N) \rceil}\right)^k.\]

  \medskip

  \textbf{Approximation gap.}
  The previous paragraphs have established that the approximation gap for $\chi(G^{[k]})$ is at~least
  \[\frac{3^k \cdot \left(1+(1-\beta)/(\beta + 3(\Delta+1) \lceil \log(3N) \rceil)\right)^k}{3^k} = \left(1+\frac{1-\beta}{\beta + 3(\Delta+1) \lceil \log(3N) \rceil}\right)^k.\]
  Thus we just need to show that \[r := n^{\gamma/(\log \log n)^2} < \left(1+\frac{1-\beta}{\beta + 3(\Delta+1) \lceil \log(3N) \rceil}\right)^k\] in order to conclude.

  \begin{claim}
    It holds that $\log r < k \log \left(1+\frac{1-\beta}{\beta + 3(\Delta+1) \lceil \log(3N) \rceil}\right)$.
  \end{claim}
  \begin{claimproof}
    We have $\log r = \frac{\gamma \log n}{(\log \log n)^2}$ and
    \[\log n = k \log |V(G)| = k \log(3N+3q) \leqslant k \log (12(\Delta+1)N \lceil\log(3N)\rceil) \leqslant 2k \log(12(\Delta+1)N),\]
    where the last inequality uses \cref{ineq:b4}.
    
    Besides, $(\log \log n)^2 \geqslant \log^2 k$.
    Thus,
    \[\log r \leqslant \frac{\gamma \cdot 2k \log (12(\Delta+1)N)}{\log^2 k} \leqslant \frac{\gamma \cdot 24(\Delta+1) k \log N}{\log^2 k} < \gamma \cdot 48(\Delta+1) \cdot \frac{N}{\log^{c+1} N \cdot \log^2 k}\]
    \[\leqslant \gamma \cdot 48(\Delta+1) \cdot \frac{N}{\log^{c+1} N \cdot 0.25 \log^2 N} = \gamma \cdot 192(\Delta+1) \cdot \frac{N}{\log^{c+3} N},\]
    where the penultimate inequality holds because of~\cref{ineq:b5}, the last inequality uses \cref{ineq:b3}, whereas the other inequalities always hold.

    On the other hand,
    \[k \log \left(1+\frac{1-\beta}{\beta + 3(\Delta+1) \lceil \log(3N) \rceil}\right) \geqslant \frac{k \cdot 0.8 (1-\beta)}{\beta + 3(\Delta+1) \lceil \log(3N) \rceil} \geqslant \frac{k \cdot 0.8 (1-\beta)}{4(\beta + 3(\Delta+1)) \log N}\]
    \[\geqslant \frac{0.8 (1-\beta)}{4(\beta + 3(\Delta+1))} \cdot \frac{N}{\log^{c+3} N},\]
    where the first inequality holds because of~\cref{ineq:b2}, and the second uses \cref{ineq:b1}.

    And we conclude since $\gamma = \frac{0.8 (1-\beta)}{4(\beta + 3(\Delta+1))} \cdot \frac{1}{192(\Delta+1)}$.
  \end{claimproof}
  This concludes the proof since, under the promise, $H$ is 3-colorable if and only if $\mathcal A$ uses at~most $r \cdot 3^k$ colors. 
\end{proof}

\section{\kis in graphs of bounded low-radius merge-width}\label{sec:w1-hard}

In this section, we reduce \gridtiling to \kis, to show the W[1]-hardness of the latter problem in graphs of bounded low-radius merge-width.
The reduction is inspired by the one of Dreier, M{\"a}hlmann, and Siebertz.~\cite[Thm.\ 2]{DMS26} and we use their notations. 
The \gridtiling problem was introduced in~\cite{Marx07}.
Given some integers $k, n \in \N$ and a~collection~$\S$ of \emph{tile} subsets $S_\pi \subseteq [n]^2$ for each \emph{grid cell} $\pi \in [k]^2$,
the goal is to decide whether there exists a~selection of tiles $s(\pi) \in S_\pi$ for all $\pi \in [k]^2$ such that:
\medskip
\begin{compactitem}
    \item The tiles $s(\pi)$ and $s(\pi_\blacktriangleright)$ have the same first component
    
    for all cells $\pi = (i,j) \in [k-1] \times [k]$ and $\pi_\blacktriangleright := (i+1,j)$; and

    \item The tiles $s(\pi)$ and $s(\pi_\triangledown)$ have the same second component
    
    for all cells $\pi = (i,j) \in [k] \times [k-1]$ and $\pi_\triangledown := (i,j+1)$.
\end{compactitem}
\medskip
It is known that \gridtiling is W[1]-hard parameterized by the grid size $k$ and unless the ETH fails it has no $f(k)n^{o(k)}$-time algorithm for any computable function $f$~\cite[Thm.\ 14.28]{Parameterized_Algorithms}. 

\begin{reptheorem}{thm:main-paramIS}
  \kis is W[1]-hard in graphs of bounded $\mw_{o(k)}$.
\end{reptheorem}
\begin{proof}
  Let $r\geq 3$ be an odd integer.
  For a~\gridtiling instance $(\S, k, n)$, we will construct a~graph $G^r_\S$ with $\mw_{r-1}(G^r_\S)\leq 4r-3$.  
  For each grid cell $\pi \in [k]^2$, we construct a cell gadget~$C^r_{\pi}$, illustrated in \cref{fig:gadgetcell}.  
  The cell gadget $C^r_{\pi}$ consists of a~horizontal path and a~vertical path of $r$ cliques:
  \begin{compactitem}
    \item $B^{h_1}_\pi,B^{h_2}_\pi, \dots, B^{h_r}_\pi $ where $B^{h_i}_\pi$ has vertex set $\{ b(\tau, h_i, \pi) \mid \tau \in S_\pi\}$; and
    \item $B^{v_1}_\pi, B^{v_2}_\pi \dots, B^{v_r}_\pi$ where $B^{v_i}_\pi$ has vertex set $\{ b(\tau, v_i, \pi) \mid \tau \in S_\pi\}$.
  \end{compactitem}
  They cross in the middle, i.e., $B^{h_m}_\pi = B^{v_m}_\pi$ for $m = \frac{r+1}{2}$.

  We then add a~co-matching between any two consecutive blocks.
  Formally, for every $\tau, \tau'\in S_\pi$ with $\tau \neq \tau'$, and every $i\in [r-1]$, we add an edge 
\begin{compactitem}
    \item between $b(\tau, h_i, \pi)$ and $b(\tau', h_{i+1}, \pi)$; and 
    \item between $b(\tau, v_i, \pi)$ and $b(\tau', v_{i+1}, \pi)$. 
  \end{compactitem}

  Let us name the endpoints of the paths after the four directions \emph{up}, \emph{down}, \emph{left}, and \emph{right}:
  \[ 
  L_\pi  = B^{h_1}_\pi \qquad 
  R_\pi = B^{h_r}_\pi \qquad 
  U_\pi = B^{v_1}_\pi \qquad 
  D_\pi = B^{v_r}_\pi  
  \]

  We call these blocks the \emph{border blocks}, and any other block in~$C^r_{\pi}$ is called an \emph{interior block}.
  We construct $G^r_\S$ by taking the disjoint union of all the $C^r_{\pi}$ for $\pi \in [k]^2$, and adding edges between border blocks of adjacent cells.
  Specifically, we add an edge between 
  \begin{itemize}%\setcounter{enumi}{2}
    \item $b(\tau, h_r, \pi) \in R_\pi$ and $b(\tau', h_1, \pi_\blacktriangleright)\in L_{\pi_\blacktriangleright}$ if and only if $\tau$ and $\tau'$ have a different first component, for all $\pi \in [k-1] \times [k]$, and  %    for all $\pi \in [k-1] \times [k]$, $\tau \in S_\pi$, $\tau' \in S_{\pi_\blacktriangleright}$; and
    
    \item $b(\tau, v_r, \pi) \in D_\pi$ and $b(\tau', v_1, \pi_\triangledown)\in U_{\pi_\triangledown}$ if and only if $\tau$ and $\tau'$ have a different second component, for all $\pi \in [k] \times [k-1]$. %    for all $\pi \in [k] \times [k-1]$, $\tau \in S_\pi$, $\tau' \in S_{\pi_\triangledown}$.
\end{itemize}

  This concludes the construction of $G^r_\S$.
  Note that $|V(G^r_\S)| = \Oh(rn^2k^2)$.

\begin{figure} 
  \centering 
  \includegraphics[width= 0.9\textwidth]{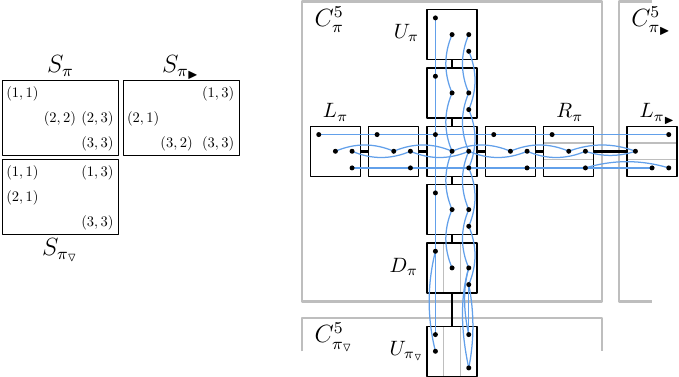}
  \caption{Construction of a gadget cell $C^5_\pi$ for a cell $\S_\pi$, and its edge-relation with the right and down gadget cells. Any set of vertices in a black box forms a clique. Two black boxes linked with a black thick edge form a biclique minus the edges in blue, which are non-edges. }\label{fig:gadgetcell}
\end{figure}

\begin{claim}\label{clm:grid_is_reduction}
    $(\S,k,n)$ is a positive \gridtiling instance if and only if $G^r_\S$ is a positive \textsc{$(2r-1)k^2$-Independent Set} instance.
\end{claim}
\begin{claimproof}
    Assume there is a~solution for $(\mathcal{S},k,n)$ using tiles $s(\pi) \in S_\pi$ for every $\pi \in [k]^2$.
    We claim that   
    \[
        I = \bigcup_{\pi \in [k]^2} I_\pi
        \quad
        \text{with}
        \quad
        I_\pi := \{b(s(\pi), h_i, \pi), b(s(\pi), v_i, \pi) \mid   i\in[r]\}
    \]
    is an independent set of size $(2r-1)k^2$.
    Clearly, $I$ has the desired size.
    For each cell $\pi$, $I_\pi$ is an independent set since it comprises only copies of the same tile $s(\pi)$.
    By construction, two vertices in distinct sets $x \in I_\pi$ and $y \in I_{\pi'}$ can only be adjacent if, up to symmetry, either
    \begin{itemize}
        \item $\pi' = \pi_\blacktriangleright$, and
        $x \in R_\pi$ and 
        $y \in L_{\pi_\blacktriangleright}$
        have a different first component; or

        \item $\pi' = \pi_\triangledown$, and
        $x \in D_\pi$ and 
        $y \in U_{\pi_\triangledown}$
        have a different second component.
    \end{itemize}
    Both possibilities are ruled out due to the $s(\pi)$ forming a solution to \textsc{Grid Tiling}. Hence, $I$ is indeed an independent set.

    Conversely, assume there exists an independent set $I$ of size $(2r-1)k^2$.
    Since we constructed the $k^2$ graphs $C^r_\pi$ from $2r-1$ cliques each, 
    $I$ must contain exactly $2r-1$ vertices from each $C^r_\pi$.
    Since those $2r-1$ vertices are pairwise non-adjacent, 
    they must be of the form
    \[
      I_\pi := \{b(s(\pi), h_i, \pi), b(s(\pi), v_i, \pi) \mid i\in[r]\}
    \]
    for some common $s(\pi) \in S_\pi$.
    We now argue that the tiles $s(\pi)$ form a \textsc{Grid Tiling} solution for $(\mathcal{S},k,n)$.
    For every $\pi \in [k-1] \times [k]$, we have that $b(s(\pi), h_r, \pi) \in R_\pi$ and $b(s(\pi)_\blacktriangleright, h_1, \pi_\blacktriangleright)\in L_{\pi_\blacktriangleright}$ are part of the same independent set $I$ and hence non-adjacent.
    By construction, we must have that $s(\pi)$ and $s(\pi_\blacktriangleright)$ agree on the first component.
    Analogously, for every $\pi \in [k] \times [k-1]$, the tiles $s(\pi)$ and $s(\pi_\triangledown)$ agree on the second component.
    This concludes the proof of the correctness of the reduction.
\end{claimproof}

We now upper-bound the radius-$(r-1)$ merge-width of the produced graphs.

\begin{claim}\label{clm:bounded_mw}
  $\mw_{r-1}(G^r_\S)\leq 4r-3$. 
\end{claim}
\begin{claimproof}
  We will construct a merge sequence for $G := G^r_\S$  of radius-$(r-1)$ width $4r-3$.
  First, we define the following sets:
  \begin{itemize}
  \item Let $N_{in}$ be the set of all the non-edges between consecutive copies of $B_\pi$ inside a cell gadget $C^r_\pi$ for all $\pi\in [k]^2$.
    These are the blue edges within the same cell gadget in~\cref{fig:gadgetcell}.
  \item Let $N_{out}$ be the set of all the non-edges between adjacent border blocks of adjacent cell gadgets in~$G$.
    These are the blue edges between two cell gadgets in~\cref{fig:gadgetcell}.
    \item Define the \emph{layers} of each border block: \\
    $L_\pi^i := \{b(\tau, h_1, \pi) \mid \tau \in S_\pi, \tau = (i,\_)\}$, \qquad 
    $R_\pi^i := \{b(\tau, h_r, \pi) \mid \tau \in S_\pi, \tau = (i,\_)\}$,\\
    $U_\pi^i := \{b(\tau, v_1, \pi) \mid \tau \in S_\pi, \tau = (\_,i)\}$,  \qquad
    $D_\pi^i := \{b(\tau, v_r, \pi) \mid \tau \in S_\pi, \tau = (\_,i)\}$.
  \end{itemize}

  We can define the following merge sequence.
  We use the letter $F$ for the set of resolved pairs, since $R$ is taken by the right border blocks.
  \begin{enumerate}
    \item Let $\P_1$ be the partition of $V(G)$ into singletons and $F_1 := N_{in}$.
    \item The parts of $\P_2$ are each interior block of each cell gadget, and for each $i \in [n]$ and $X \in \{L,R,U,D\}$, each layer $X^i_\pi$.
      We keep $F_2 := F_1$.
    \item The parts of~$\P_3$ are each block of each cell gadget. We set $F_3 := F_2 \cup N_{out}$.
    \item We finally let $\P_4 := \{V(G)\}$ and $F_4 := F_3 \cup E(G)$.
  \end{enumerate}
  By construction $\P_i$ is coarser than $\P_{i-1}$, $F_i\subseteq F_{i+1}$, and $\P_i$ is homogeneous modulo $F_i$.
  We now argue that the radius-$(r-1)$ width of the sequence \[(\P_1,F_1), (\P_2,F_2), (\P_3,F_3), (\P_4,F_4)\] is at most $4r-3$.

  The radius-$(r-1)$ (in fact the radius-$\infty$) width of $(\P_1,F_2=N_{in})$ is at most $2r-1$ since every connected component of $(V(G),F_2)$ has exactly $2r-1$ vertices.

  %For $(\P_1,F_2)$, since $F_2 = F_1 = N_{in}$, for any $v\in V(G)$, in the cell gadget $C^r_\pi$ there are exactly $2r-1$ vertices in $\Ball^{r-1}_{F_2}(v)$, that is, $v$ plus the $2r-2$ other copies of the same tile of $S_\pi$.

  For $(\P_2,F_3)$, observe that for any vertices $x$ and $y$ in the same border block $B_\pi$ of $C^r_\pi$ but in different layers, $\dist_{F_3}(x,y)\geq 2r$.
  Indeed, if there is a path of length at most $2r$ between them, then it must be of the form 
  $x,\dots,x',z,y',\dots,y$ where $x'$ and $y'$ are the respective copies of $x$ and $y$ in another border block $B'_\pi$ of $C^r_\pi$, and $z$ is a vertex in the border block adjacent to $B'_\pi$ such that $x'z,y'z \in N_{out}$. 
  Since $\dist_{F_3}(x,x')= \dist_{F_3}(x,y')=r-1$, it follows that $\dist_{F_3}(x,y)\geq 2r$.
  Therefore, for any $v\in V(G)$, the set $\Ball^{r-1}_{F_3}(v)$ intersects at most one layer from each border block. 
  Since each layer of a border block and each interior block is in one part of $\P_2$, it follows that for any $v\in V(G)$, we have $|\Ball^{r-1}_{F_3}(v)/\P_2| \leq 4r-3$.

  Finally, for $(\P_3,F_4)$, the edges of $G$ are only between consecutive blocks inside a cell gadget, or between adjacent border blocks. Since any block of $G$ is in one part of $\P_3$, for any $v\in V(G)$, we have $|\Ball^{r-1}_{F_4}(v)/\P_3| \leq 4r-3$. This concludes the proof.
\end{claimproof}

These claims readily give that \textsc{$k$-IS} is W[1]-hard on graphs of bounded radius-$r$ merge-width, where $r$ is a~fixed constant.
However, we can extract more. %: For any functions $f(k) = o(k)$ and $g(k)$ with $g(k) \ge 4k$, solving \textsc{$k$-IS} on graphs of $\mw_{f(k)} \le g(k)$ is W[1]-hard.
Fix two computable functions $f(x) = o(x)$ and $g(x) \ge 4x$.
Then we say that $G$ has \emph{$(f, g, k)$-merge-width} if $\mw_{f(k)}(G) \le g(k)$.
%Consider the problem \textsc{$k$-IS on $(f, g)$-merge-width} that takes as input a parameter $k$, and a graph $G$ such that $\mw_{f(k)}(G) \le g(k)$.
\begin{claim}\label{clm:w1-hardness}
  \kis on graphs with $(f, g, k)$-merge-width is W[1]-hard.
\end{claim}
\begin{claimproof}
%We show an FPT reduction from \gridtiling to \textsc{$k$-IS on $(f, g)$-merge-width}.
  %For any integer $k$,
  We design an FPT reduction from \gridtiling with parameter $q$ to \kis with $k := (2r - 1)q^2$, where $r := r(q)$ is the smallest odd integer at~least equal~to $3$ such that \[f((2r-1)q^2) \le r-1.\]
This integer exists since $f(x) = o(x)$, and is computable since $f$ is computable.

Let $(\S,q,n)$ be an instance of \gridtiling.
Computing $G := G^r_{\S}$ can be done in time $\Oh(\lVert G \rVert) = \Oh(r^2q^4 \cdot n^4)$.
This is an FPT reduction in parameter $q$ since $r$ depends only on~$q$.
Besides, $G$ has $(f, g, k)$-merge-width, since 
\[
\mw_{f(k)}(G) = \mw_{f((2r - 1)q^2)}(G)
\le \mw_{r - 1}(G)
\le 4r - 3 
\le 4(2r - 1)q^2
\le g((2r - 1)q^2) = g(k).
\]
Where the first inequality comes from monotonicity of merge-width with respect to the radius, the second is due to \Cref{clm:bounded_mw}, and the last, by the assumption on~$g$.
We conclude by~\cref{clm:grid_is_reduction}.
%Moreover, by \Cref{clm:grid_is_reduction} $(\S,q,n)$ is a positive instance of \gridtiling if and only if $G$ is a positive instance of \kis on graphs with $(f, g, k)$-merge-width.
\end{claimproof}

In particular, for every $r=o(k)$, \textsc{$k$-IS} is W[1]-hard on graphs of bounded $\mw_r$.
\end{proof}

\section{\kds in graphs of bounded low-radius merge-width}\label{sec:w1-hardDS}

In this section, we adapt the reduction from the previous section to show the W[1]-hardness of \kds in graphs of low-radius merge-width.
The main difference is how the selection of copies of the same tile in a~cell gadget is enforced. 
First, instead of cliques, the blocks will be independent sets. Then, we will achieve the vertex selection by inserting between every pair of consecutive blocks $B$ and $B'$, an additional block $A$ such that there is a co-matching between $B$ and $A$, and a matching between $A$ and $B'$. The only way to dominate this new block $A$ with two vertices will be to pick copies of the same tile in $B$ and $B'$. A similar construction can be done to enforce that selected tiles in adjacent cells agree on their corresponding component (see \cref{fig:gadgetcell_ds}).

\begin{reptheorem}{thm:main-paramDS}
  \kds is W[1]-hard in graphs of bounded $\mw_{o(k)}$.
\end{reptheorem}

\begin{proof}
  Let $r>3$ be an odd integer, for a \gridtiling instance $(\S, k, n)$ we will construct a graph $G^r_\S$ with $\mw_{2r-2}(G^r_\S)= \Oh(r)$.  
  For each grid cell $\pi \in [k]^2$, we construct a cell gadget $C^r_{\pi}$.  
  The cell gadget $C^r_{\pi}$ consists of a~horizontal path and a~vertical path of $2r-1$ edgeless graphs:
  \begin{compactitem}
    \item $B^{h_1}_\pi,B^{h_2}_\pi, \dots, B^{h_{2r-1}}_\pi $ where $B^{h_i}_\pi$ have  vertex set $\{ b(\tau, h_i, \pi) \mid \tau \in S_\pi\}$; and
    \item $B^{v_1}_\pi, B^{v_2}_\pi, \dots, B^{v_{2r-1}}_\pi$ where $B^{v_i}_\pi$ have  vertex set $\{ b(\tau, v_i, \pi) \mid \tau \in S_\pi\}$.
  \end{compactitem}
  They are crossing in the middle, i.e., $B^{h_r}_\pi = B^{v_r}_\pi$.

  For every even block we add a co-matching towards the previous block and a matching towards the next block.
  That is, for $\tau, \tau'\in S_\pi$ with $\tau \neq \tau'$ and even $i\in[2r-1]$, add to $C^r_{\pi}$ edges between 
  \begin{itemize}
    \item $b(\tau, h_i, \pi)$ and $b(\tau', h_{i-1}, \pi)$; \quad $b(\tau, h_i, \pi)$ and $b(\tau, h_{i+1}, \pi)$;  
    \item $b(\tau, v_i, \pi)$ and $b(\tau', v_{i-1}, \pi)$; \quad $b(\tau, v_i, \pi)$ and $b(\tau, v_{i+1}, \pi)$. 
  \end{itemize}

  Additionally, for each odd $i\in[2r-1]$, add 
  \begin{itemize}
    \item vertices $a_1(h_i,\pi)$ and $a_2(h_i,\pi)$ complete to $B^{h_i}_\pi$; and
    \item vertices $a_1(v_i,\pi)$ and $a_2(v_i,\pi)$ complete to $B^{v_i}_\pi$.
  \end{itemize}

  Note that since $r$ is odd, a pair of such vertices is created for $B^{h_r}_\pi = B^{v_r}_\pi$. This is done to guarantee that a vertex is picked in each odd block.
  The endpoints of the paths are called \emph{border blocks}, and any other block in~$C^r_{\pi}$ is called an \emph{interior block}.
  
  % after the four directions \emph{up}, \emph{right}, \emph{down}, and \emph{left}:
  % \[ 
  % U_\pi = B^{h_1}_\pi \qquad 
  % D_\pi = B^{h_{2r-1}}_\pi  \qquad 
  % L_\pi  = B^{v_1}_\pi \qquad 
  % R_\pi = B^{v_{2r-1}}_\pi
  % \]

  We construct $G^r_\S$ by taking the disjoint union of all the $C^r_{\pi}$ for $\pi \in [k]^2$.
  Additionally, to ensure that the same first or second component is selected between adjacent gadget cells, we add the following vertices and edges.
  \begin{itemize}      
      \item For each $\pi \in [k-1] \times [k]$, 
  add a set of vertices $\{r(i,\pi)\mid i\in [n]\}$ and edges 
  between  $b((i,\_), h_{2r-1}, \pi) \in B^{h_{2r-1}}_\pi$ 
  and $r(j,\pi)$ for $i, j\in [n]$ with $i \neq j$; 
  and $b((i,\_), h_{1}, \pi_\blacktriangleright) \in B^{h_1}_{\pi_\blacktriangleright}$ and $r(i,\pi)$ for $i\in [n]$.
   
  \item For each $\pi \in [k] \times [k-1]$, 
  add a set of vertices $\{d(i,\pi)\mid i\in [n]\}$ and edges 
  between  $b((\_,i), v_{2r-1}, \pi) \in B^{v_{2r-1}}_\pi$ and $d(j,\pi)$ for $i\neq j\in [n]$; 
  and $b((\_,i), v_{1}, \pi_\triangledown) \in B^{v_{1}}_{\pi_\triangledown}$ and $d(i,\pi)$ for $i\in [n]$.
    \end{itemize}
This concludes the construction of $G^r_\S$ and note that $|V(G^r_\S)| = \Oh(rn^2k^2)$.

\begin{figure} 
  \centering 
  \includegraphics[width= 0.8\textwidth]{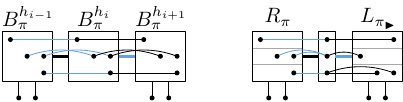}
  \caption{Left: Co-matching and matching between consecutive blocks in a gadget cell with $i$ even. Right: construction between two border blocks of adjacent cells.
  Any set of vertices in a black box forms an edgeless graph. Vertices incident to a~black box are complete to the set of vertices contained in that box.
  }\label{fig:gadgetcell_ds}
\end{figure}

The following claims can be shown similarly to their counterparts in \cref{thm:main-paramIS}.
For the equivalence of \gridtiling and \kds instances, it follows from the fact that the only dominating sets of size $2r-1$ of a gadget cell $C_\pi^r$ are of the form
  \[
        \text{Dom}_\pi := 
        \{b(\tau, h_{2i-1}, \pi), b(\tau, v_{2i-1}, \pi) \mid   i\in[r]\}
    \]
   for  some $\tau \in S_\pi$. Moreover, the vertices added between adjacent gadget cells force any dominating set to select the corresponding component in a~dominating set of size $(2r-1)k^2$.

\begin{claim}\label{clm:grid_ds_reduction}
    $(\S,k,n)$ is a positive \gridtiling instance if and only if 
    $G^r_\S$ is a positive $(2r-1)k^2$-DS instance.
\end{claim}

Then, for the bound on the radius-$(2r-2)$ merge-width, the merge sequence is done similarly: first resolve all matchings and co-matchings inside gadget cells; then resolve the co-matchings and the matchings with the vertices between adjacent gadget cells; then add every remaining edge of $G^r_\S$.

\begin{claim}\label{clm:bounded_mw_ds}
  $\mw_{2r-2}(G^r_\S)= \Oh(r)$.
\end{claim}

Using the same arguments as in~\cref{clm:w1-hardness}, it follows that for any functions $f(k) = o(k)$ and $g(k)$, solving \kds on graphs with $\mw_{f(k)} \le g(k)$ is W[1]-hard.
\end{proof}

\end{document}